\documentclass[12pt]{article}
\usepackage{amsthm}
\newtheorem{definition}{Definition}
\newtheorem{proposition}{Proposition}
\newtheorem{theorem}{Theorem}
\newtheorem{lemma}{Lemma}
\newtheorem{assumption}{Assumption}
\newtheorem{corollary}{Corollary}
\newtheorem{remark}{Remark}
\newtheorem{assumptionA}{Assumption A}
\usepackage{sourceserifpro}
\usepackage{microtype}
\usepackage{ragged2e}
\usepackage{amsmath}
\usepackage{parskip}
\usepackage{subcaption}
\usepackage{amsfonts}
\usepackage{rotating}
\usepackage{threeparttable}
\usepackage{array}
\usepackage{graphicx}
\usepackage{caption}
\usepackage{multirow}
\usepackage{booktabs}
\usepackage{amssymb}
\usepackage{setspace}
\usepackage{natbib}
\usepackage{hyperref}
\hypersetup{
  colorlinks=true,
  linkcolor=red,
  citecolor=red,
  urlcolor=red,
  pdfborder={0 0 0}
}
\usepackage[export]{adjustbox}
\usepackage{wrapfig}
\usepackage{csquotes}
\bibpunct{(}{)}{;}{a}{,}{,}

\newenvironment{keywords}{\par\vspace{0.5em}\noindent\textbf{Keywords:} }{\par}
\newenvironment{jel}{\par\vspace{0.25em}\noindent\textbf{JEL Classification:} }{\par}

\newcommand{\Imed}{\mathcal{I}}
\newcommand{\Ired}{\mathcal{I}^{(-X)}}
\newcommand{\Menu}{\mathcal{M}}

\newcommand{\Eb}{\mathbb{E}}
\newcommand{\Pb}{\mathbb{P}}
\newcommand{\Rb}{\mathbb{R}}

\title{\textbf{Distributional Granger Causality: Identification, Sequential Inference, and Adaptive Testing}}

\author{Ayush Jha\footnote{Corresponding author: Ayush Jha; \tt{ayush.jha@ttu.edu}} \\
\small Department of Economics, Texas Tech University, Lubbock, TX, USA \
}
\date{}

\begin{document}

\maketitle

\begin{abstract}
Predictive dependence in time series need not be confined to the conditional mean. Outside the Gaussian setting, causal content may arise through conditional scale, tail behavior, asymmetry, or other distributional features, implying that no single Granger-type test provides a complete characterization of predictive dependence. This paper develops a framework for distributional Granger causality based on a finite collection of channel-specific restrictions. Under suitable determinacy conditions, the channel menu is shown to be complete, yielding an identification result that links distributional Granger non-causality to a finite set of testable hypotheses. Building on this representation, we develop an adaptive sequential testing procedure that allocates inferential resources across channels while maintaining familywise error control through an alpha-investing mechanism. A policy-invariant validity theorem establishes finite-sample size control under arbitrary admissible selection rules, while an asymptotic efficiency theorem shows that a confidence-bound allocation rule achieves power equivalent to that of an infeasible oracle benchmark. The theoretical guarantees are derived from primitive mixing and moment conditions together with a circular-block permutation scheme.
\end{abstract}


\begin{jel}
C12; C14; C22; C52.
\end{jel}

\onehalfspacing

\section{Introduction}\label{sec:introduction}

Granger causality is among the most widely used tools for studying predictive relationships in time series. In its classical form, the concept is operationalized through incremental forecasting ability: a process $X_t$ is said to Granger-cause a process $Y_t$ if past values of $X_t$ improve forecasts of $Y_t$ beyond the information contained in the past of $Y_t$ alone \citep{Granger1969,Sims1972}. Under linear dynamics and Gaussian innovations, this notion admits a particularly simple characterization through lag-exclusion restrictions in vector autoregressions \citep{Geweke1982,Geweke1984}. In such environments, the conditional distribution is fully summarized by its first two moments, and causal inference reduces to a problem of conditional-mean predictability.

Outside the Gaussian setting, however, predictive dependence need not be confined to the conditional mean. A predictor may influence the conditional variance of future outcomes, alter tail probabilities, modify higher-order moments, or affect other features of the conditional distribution without changing expected values. This observation has motivated a large literature on alternative notions of Granger causality, including causality in variance \citep{CheungNg1996}, causality in risk and tail events \citep{HongLiuWang2009,WhiteKimManganelli2015}, causality in quantiles \citep{JeongHardleSong2012,SongTaamouti2021}, nonlinear and nonparametric causality \citep{HiemstraJones1994,DiksPanchenko2006,NishiyamaHitomiKawasakiJeong2011}, copula-based dependence measures \citep{BouezmarniRomboutsTaamouti2012}, and information-theoretic approaches such as transfer entropy \citep{BarnettSethBossomaier2009}.

The resulting literature provides a rich collection of channel-specific tests but leaves open a fundamental question. If predictive dependence may arise through multiple dimensions of the conditional distribution, how should inference be conducted when the relevant channel is unknown ex ante? In empirical practice, researchers often evaluate several causality tests and interpret the resulting collection of $p$-values informally. This approach faces two difficulties. First, simultaneous consideration of multiple channels generates a multiple-testing problem, potentially leading to substantial distortions in familywise error rates. Second, there is generally no principled rule for allocating finite inferential resources across competing tests or for determining which channel should receive the greatest attention.

This paper develops a unified framework for distributional Granger causality that addresses both issues. The starting point is an identification result. Rather than viewing existing causality tests as competing methodologies, the paper interprets them as measuring distinct coordinates of a common object: the conditional distribution of future outcomes. Distributional Granger non-causality is characterized through a finite collection of channel-specific restrictions corresponding to conditional location, scale, tail behavior, and higher-order distributional features. Under suitable determinacy conditions, these restrictions are shown to be complete in the sense that distributional Granger non-causality holds if and only if every channel-specific restriction is satisfied. This representation transforms an infinite-dimensional hypothesis concerning conditional distributions into a finite collection of testable components while preserving identification of the underlying causal object.

Building on this representation, the paper develops an adaptive sequential testing procedure for channel selection. The procedure allocates a finite testing budget across channels, updates allocation decisions using previously observed outcomes, and terminates once sufficient evidence against the null has been accumulated or the testing budget has been exhausted. The inferential framework combines channel-specific hypothesis tests with an alpha-investing mechanism drawn from the sequential multiple-testing literature \citep{FosterStine2008,AharoniRosset2014}. This construction permits data-dependent channel selection while maintaining rigorous control of familywise error rates.

The theoretical contribution consists of three results. First, a completeness theorem establishes that the proposed channel menu fully characterizes distributional Granger non-causality. Second, a policy-invariant familywise error theorem shows that inferential validity is preserved under arbitrary admissible channel-selection rules. The result separates validity from allocation: any adaptive policy satisfying the filtration-adaptedness requirements inherits the same finite-sample error guarantee. Third, an asymptotic efficiency theorem demonstrates that a confidence-bound allocation rule achieves power equivalent to that of an infeasible oracle benchmark in the limit. Consequently, the power loss associated with uncertainty regarding the active channel vanishes asymptotically.

An important feature of the analysis is that the theoretical guarantees are derived from primitive conditions on the data-generating process rather than imposed as abstract assumptions. Conditional super-uniformity of the channel-specific $p$-values follows from a circular-block permutation procedure under strict stationarity and suitable mixing conditions. Likewise, the concentration properties required for the efficiency analysis are obtained from Bernstein-type inequalities for dependent processes. These results provide an explicit link between the high-level assumptions underlying the sequential testing framework and standard regularity conditions for weakly dependent time series.

The framework has implications for several areas of applied econometrics. In financial applications, predictive dependence frequently appears through volatility, downside risk, or tail exposure rather than through expected returns. Similar considerations arise in macroeconomics, network analysis, and systemic-risk measurement, where distributional spillovers are often more informative than conditional-mean effects. By treating causality as a property of the entire conditional distribution rather than a single moment condition, the proposed framework accommodates such settings while preserving formal inferential guarantees.

The remainder of the paper is organized as follows. Section~\ref{sec:literature} reviews the related literature. Section~\ref{sec:framework} develops the distributional-causality framework and establishes the completeness of the channel representation. Section~\ref{sec:agent} introduces the adaptive sequential testing procedure and presents the validity and efficiency results. Section~\ref{sec:mc} reports Monte Carlo evidence on finite-sample performance. Section~\ref{sec:conclusion} concludes. Proofs and additional technical results are collected in the Online Appendix.

\section{Related Literature}\label{sec:literature}

This paper contributes to three strands of the econometrics literature: distributional approaches to Granger causality, sequential multiple-testing procedures with error-control guarantees, and adaptive allocation methods for statistical experimentation.

\subsection{Distributional Granger causality}

The classical formulation of Granger causality characterizes predictive dependence through incremental linear forecasting ability \citep{Granger1969,Sims1972,Geweke1982,Geweke1984}. Under Gaussianity and linear dynamics, this characterization is sufficient because the conditional distribution is fully summarized by its first two moments. Outside that setting, however, predictive content may arise through conditional scale, tail behavior, asymmetry, or other higher-order distributional features that are not captured by conditional-mean restrictions alone.

A substantial literature has therefore developed tests targeting specific dimensions of predictive dependence. Examples include causality in variance \citep{CheungNg1996}, causality in risk and tail events \citep{HongLiuWang2009,WhiteKimManganelli2015}, causality in distribution \citep{CandelonTokpavi2016}, causality in quantiles \citep{JeongHardleSong2012,SongTaamouti2021}, nonlinear and nonparametric causality \citep{HiemstraJones1994,DiksPanchenko2006,NishiyamaHitomiKawasakiJeong2011}, and copula-based approaches to conditional dependence \citep{BouezmarniRomboutsTaamouti2012}. Information-theoretic measures such as transfer entropy provide additional tools for detecting departures from linear predictability and coincide with classical Granger causality under Gaussianity \citep{BarnettSethBossomaier2009}.

Related developments arise in the literature on identification under non-Gaussianity, where higher-order moments and distributional features play a central role in recovering structural relationships \citep{ShimizuHoyerHyvarinenKerminen2006,LanneMeitzSaikkonen2017,GourierouxMonfortRenne2017,MontielOleaPlagborgMollerQian2022}. Higher-order cumulants likewise provide a natural representation of departures from Gaussianity and characterize important dimensions of distributional shape \citep{ArevalilloNavarro2026}. Recent work has also continued to refine inference for predictive regressions and Granger-causality tests in nonstandard environments, including settings involving boundary parameters and weak identification \citep{CavaliereGeorgiev2020,CavaliereGeorgievZanelli2025}.

The present paper differs from this literature in its objective. Rather than proposing a new channel-specific test, it studies how existing tests can be combined within a unified framework. The contribution is an identification result showing that, under suitable determinacy conditions, a sufficiently rich collection of channel-specific restrictions provides a complete characterization of distributional Granger non-causality.

\subsection{Sequential multiple testing and online inference}

Testing multiple dimensions of predictive dependence naturally raises a multiple-testing problem. Classical approaches control familywise error rates through fixed corrections such as Bonferroni procedures, but these methods are static and often conservative. A complementary literature studies sequential and online testing procedures that allocate significance levels adaptively as hypotheses are examined \citep{FosterStine2008,AharoniRosset2014,JavanmardMontanari2018,RamdasZrnicWainwrightJordan2018}.

Closely related developments arise in the literature on anytime-valid inference, test martingales, and e-processes, which provide inferential guarantees that remain valid under optional stopping and adaptive continuation rules \citep{RamdasGrunwaldVovkShafer2023,GrunwaldHeideKoolen2024}. These methods emphasize the construction of inferential procedures whose validity is preserved under data-dependent experimentation.

The present paper applies these ideas to distributional causality testing. The distinguishing feature of the setting considered here is that the sequence of hypotheses is not exogenously given but is selected adaptively from a menu of candidate channels. The resulting policy-invariant familywise error theorem establishes that valid inference is preserved under arbitrary admissible channel-selection rules. In this sense, the paper extends sequential testing methods to an environment in which both testing and selection are data dependent.

\subsection{Adaptive allocation and statistical efficiency}

The allocation of finite testing resources across competing hypotheses is closely related to the literature on sequential experimentation and adaptive allocation. Multi-armed bandit methods provide a canonical framework for balancing exploration and exploitation and have generated a large body of results concerning regret minimization and efficient learning \citep{AuerCesaBianchiFischer2002,LattimoreSzepesvari2020,KaufmannCappeGarivier2016}.

The problem considered here differs from the standard bandit setting in both objective and interpretation. The goal is not reward maximization per se, but efficient allocation of inferential effort across competing dimensions of predictive dependence. The objects being learned are channel-specific measures of statistical informativeness, represented through the non-centrality parameters of the underlying tests.

This perspective connects with a broader movement toward adaptive and data-driven procedures in econometrics \citep{ChernozhukovEtAl2018,FarrellLiangMisra2021}. Recent methodological discussions have emphasized the growing interaction between econometric theory, machine learning, and artificial intelligence \citep{GuggenbergerSuSun2026}. The contribution of the present paper is to embed adaptive allocation within a formal inferential framework. The resulting procedure combines three features that are typically studied separately: identification of the causal object, finite-sample error control, and asymptotically efficient allocation of testing resources. The asymptotic oracle-efficiency result demonstrates that adaptation can improve power without compromising the validity guarantees established by the sequential testing framework.

\section{Framework: Channel Decomposition and Menu Completeness}\label{sec:framework}

\subsection{Notation and information sets}

Let $\{(X_t,Y_t)\}_{t\in\mathbb{Z}}$ be a strictly stationary bivariate stochastic process defined on the probability space $(\Omega,\mathcal{F},\mathbb{P})$. Define the information sets
\[
\Imed_{t-1} := \sigma\big(\{X_{t-j},Y_{t-j}\}_{j\ge 1}\big),
\qquad
\Ired_{t-1} := \sigma\big(\{Y_{t-j}\}_{j\ge 1}\big).
\]
Let $F_{Y_t\mid\mathcal{J}}(\cdot)$ denote a regular conditional distribution of $Y_t$ given a sub-$\sigma$-field $\mathcal{J}$, and let $Q_{Y_t}(\tau\mid\mathcal{J})$ denote the associated conditional $\tau$-quantile. Unless otherwise stated, all equalities and inequalities involving conditional objects hold $\Pb$-almost surely, with continuity-point qualifications imposed where required.

The subsequent asymptotic analysis relies on a collection of primitive conditions governing dependence, moments, and conditional smoothness. These assumptions are standard in the nonparametric and resampling literature for weakly dependent stochastic processes and are sufficient to establish the higher-level regularity conditions employed in the finite-sample size and power results developed in Section~\ref{sec:agent}. Formal derivations are provided in Appendix~\ref{OA:primitives}.

\begin{assumption}[Data-generating primitives]\label{ass:primitives}
The process $\{(X_t,Y_t)\}_{t\in\mathbb{Z}}$ is strictly stationary and absolutely regular ($\beta$-mixing) with coefficients $\beta(m)\le c_0\,\rho^{m}$ for some $c_0<\infty$, $\rho\in(0,1)$ (geometric mixing; polynomial decay $\beta(m)=O(m^{-b})$ with $b>2$ suffices for all results at the stated rates). The variables admit $4+\epsilon$ finite moments, $\Eb\,|Y_t|^{4+\epsilon}+\Eb\,|X_t|^{4+\epsilon}<\infty$ for some $\epsilon>0$. The conditional distribution $F_{Y_t\mid\Imed_{t-1}}$ has a density that is bounded and, near each tail quantile $Q_{Y_t}(\tau_L\mid\cdot)$, $Q_{Y_t}(\tau_U\mid\cdot)$, bounded away from zero and Lipschitz in $y$ uniformly in the conditioning history.
\end{assumption}

\begin{definition}[Circular-block permutation scheme]\label{def:permscheme}
For a channel statistic $S_k$ computed from $\{(X_t,Y_t)\}_{t=1}^T$, define the permutation $p$-value
\[
P_k =
\frac{
1+\#\{b:S_k(X^{(b)},Y)\ge S_k(X,Y)\}
}{B+1},
\]
where each permuted sequence is generated by the circular shift
\[
X^{(b)}_t=X_{((t+s_b-1)\bmod T)+1},
\]
with $s_b$ drawn independently and uniformly from $\{1,\ldots,T-1\}$ for $b=1,\ldots,B$.

The circular shift preserves the marginal distribution and serial dependence structure of $\{X_t\}$. Under the channel null hypothesis $H_{0,k}$, the resulting statistic is invariant in distribution to such shifts, providing the basis for permutation validity; see Appendix~\ref{OA:primitives}.
\end{definition}

\subsection{Distributional Granger non-causality}

The object of interest is the conditional distribution of $Y_t$ given the available information. The strongest form of predictive irrelevance is defined through equality of conditional laws.

\begin{definition}[Granger non-causality in distribution]\label{def:gc_dist}
The process $X$ does not Granger-cause $Y$ in distribution if
\[
F_{Y_t\mid\Imed_{t-1}}(y) = F_{Y_t\mid\Ired_{t-1}}(y)
\quad
\forall y\in\Rb,\ \forall t.
\]
\end{definition}

Definition~\ref{def:gc_dist} characterizes predictive irrelevance at the level of the entire conditional distribution. It therefore excludes any incremental predictive contribution of the history of $X$ once the history of $Y$ has been conditioned upon. Since this condition is infinite-dimensional, direct implementation is generally infeasible. The framework developed below replaces the distributional null by a collection of lower-dimensional restrictions defined through measurable functionals of the conditional law.

\subsection{Channel functionals and channel nulls}

Let $\Menu=\{1,\ldots,K\}$ index a finite collection of measurable functionals of the conditional distribution. For each channel $k\in\Menu$, let $\phi_k$ denote a functional defined on the space of conditional laws. The corresponding channel null hypothesis is
\[
H_{0,k}:
\quad
\phi_k\!\left(F_{Y_t\mid\Imed_{t-1}}\right) = \phi_k\!\left(F_{Y_t\mid\Ired_{t-1}}\right)
\quad
\Pb\text{-a.s.},
\qquad \forall t.
\]

The canonical menu considered throughout the paper is given by
\begin{align}
\phi_1 &= \Eb[Y_t\mid\cdot]
&&\text{(conditional mean),}\notag\\
\phi_2 &= \mathrm{Var}(Y_t\mid\cdot)
&&\text{(conditional variance),}\notag\\
\phi_3 &= Q_{Y_t}(\tau_L\mid\cdot)
&&\text{(lower-tail quantile),}\notag\\
\phi_4 &= Q_{Y_t}(\tau_U\mid\cdot)
&&\text{(upper-tail quantile),}\notag\\
\phi_5 &= \big(\kappa_3,\kappa_4\big)(Y_t\mid\cdot)
&&\text{(higher-order conditional cumulants).}
\label{eq:menu}
\end{align}

The first four channels correspond to conditional location, scale, and tail behavior. The fifth channel captures higher-order distributional features through conditional skewness- and kurtosis-related cumulants. Each channel induces a testable restriction and may be associated with an established inferential procedure. Throughout, the testing methodology is treated as given; the emphasis is on the logical relationship between the collection of channel restrictions and distributional Granger non-causality.

\subsection{The nesting property}

The channel hypotheses are nested within the distributional null.

\begin{proposition}[Distributional non-causality implies every channel null]\label{prop:dist_implies_channels}
Suppose the functionals $\phi_k$ in \eqref{eq:menu} are well defined. If $X$ does not Granger-cause $Y$ in distribution according to Definition~\ref{def:gc_dist}, then $H_{0,k}$ holds for every $k\in\Menu$.
\end{proposition}

\begin{proof}
Under Definition~\ref{def:gc_dist},
\[
F_{Y_t\mid\Imed_{t-1}} = F_{Y_t\mid\Ired_{t-1}}
\qquad
\Pb\text{-a.s.}
\]
for every $t$. Since each $\phi_k$ is a measurable functional of the conditional distribution, application of $\phi_k$ to both sides yields
\[
\phi_k\!\left(F_{Y_t\mid\Imed_{t-1}}\right) = \phi_k\!\left(F_{Y_t\mid\Ired_{t-1}}\right),
\]
which establishes $H_{0,k}$.
\end{proof}

Proposition~\ref{prop:dist_implies_channels} establishes that distributional Granger non-causality implies the validity of every channel restriction. Consequently, rejection of any channel null is sufficient to reject distributional non-causality. The converse implication, however, requires additional structure. Specifically, the collection of channel functionals must be sufficiently informative to identify the conditional law.

\subsection{Menu completeness}

To establish the converse implication, the collection of channel functionals must uniquely characterize the conditional law within an admissible class of distributions. The following assumption formalizes this requirement.

\begin{assumption}[Determinacy class]\label{ass:determinacy}
For each $t$, the conditional law $F_{Y_t\mid\Imed_{t-1}}$ belongs $\Pb$-a.s.\ to a family $\mathcal{D}$ satisfying the following properties:

\begin{enumerate}
\item[(i)] Every distribution in $\mathcal{D}$ is uniquely determined by its cumulant sequence together with its lower- and upper-tail quantile functions;

\item[(ii)] $\mathcal{D}$ is closed under conditioning.
\end{enumerate}

The skewed scale-mixture families commonly employed in financial econometrics satisfy these requirements over the parameter regions considered in empirical applications.
\end{assumption}

Assumption~\ref{ass:determinacy} is an identification condition. It guarantees that equality of the functionals defining the menu implies equality of the underlying conditional distributions. Consequently, the finite collection of channel restrictions can be used to recover an infinite-dimensional statement concerning conditional laws.

\begin{theorem}[Menu completeness]\label{thm:completeness}
Under Assumption~\ref{ass:determinacy}, suppose that the canonical menu \eqref{eq:menu} is augmented so that $\phi_5$ indexes the full conditional cumulant sequence $\{\kappa_m\}_{m\ge3}$. Then the following statements are equivalent:

\begin{enumerate}
\item[(a)] $X$ does not Granger-cause $Y$ in distribution;

\item[(b)] $H_{0,k}$ holds for every channel $k\in\Menu$.
\end{enumerate}

Consequently,
\[
X \text{ Granger-causes } Y \text{ in distribution}
\quad\Longleftrightarrow\quad
\exists\, k\in\Menu \text{ such that } H_{0,k} \text{ fails}.
\]
\end{theorem}

\begin{proof}
The implication $(a)\Rightarrow(b)$ follows directly from Proposition~\ref{prop:dist_implies_channels}.

To establish $(b)\Rightarrow(a)$, suppose that $H_{0,k}$ holds for every channel $k\in\Menu$. Then the conditional mean, conditional variance, all higher-order conditional cumulants, and the lower- and upper-tail conditional quantiles are invariant to the inclusion of the history of $X$ once the history of $Y$ has been conditioned upon.

By Assumption~\ref{ass:determinacy}(i), members of $\mathcal{D}$ are uniquely characterized by precisely these objects. Therefore the conditional laws
\[
F_{Y_t\mid\Imed_{t-1}}
\quad\text{and}\quad
F_{Y_t\mid\Ired_{t-1}}
\]
coincide $\Pb$-a.s. Since this equality holds for every $t$, Definition~\ref{def:gc_dist} follows. The final statement is obtained by contraposition.
\end{proof}

Theorem~\ref{thm:completeness} establishes that the collection of channel restrictions constitutes a complete representation of distributional Granger non-causality within the determinacy class $\mathcal{D}$. The result is fundamentally one of identification: the conditional law can be recovered from the collection of channel coordinates, and therefore testing the complete menu is equivalent to testing the distributional null itself.

\begin{remark}[Finite truncation]\label{rem:truncation}
The practical implementation of the procedure employs a truncated menu consisting of conditional scale, lower- and upper-tail quantiles, and the third and fourth conditional cumulants. The resulting collection of restrictions no longer provides a complete characterization of the conditional law.

Accordingly, failure to reject all channel nulls need not imply distributional non-causality. The finite menu instead defines a restricted null hypothesis corresponding to invariance of the selected coordinates. All finite-sample and asymptotic guarantees developed in Section~\ref{sec:agent} are stated relative to this operational null. Theorem~\ref{thm:completeness} characterizes the limiting case in which the menu is sufficiently rich to recover the full conditional distribution.
\end{remark}

\subsection{The Gaussian boundary as a degenerate menu}

The classical theory of linear Granger causality emerges as a special case in which the distributional menu collapses to a single informative coordinate.

\begin{proposition}[Gaussian collapse of the menu]\label{prop:gauss_collapse}
Suppose that the conditional law of $Y_t$ given $\Imed_{t-1}$ is Gaussian, with conditional mean affine in the conditioning history and conditional variance invariant to the conditioning history. Then:

\begin{enumerate}
\item[(i)] Channels $2$--$5$ are non-informative with respect to Granger causality;

\item[(ii)] The only potentially informative channel is the conditional mean channel $\phi_1$;

\item[(iii)] The channel null $H_{0,1}$ is equivalent to the linear lag-exclusion restrictions defining classical Granger non-causality.
\end{enumerate}

Consequently, distributional Granger non-causality, joint validity of the channel nulls, and linear Granger non-causality coincide.
\end{proposition}

\begin{proof}
Under the stated assumptions,
\[
Y_t = \mu_t + \varepsilon_t,
\]
where $\varepsilon_t$ is conditionally Gaussian with variance independent of the conditioning history. Hence the conditional variance is invariant by construction and all conditional cumulants of order $m\ge3$ vanish identically.

It follows immediately that channels corresponding to conditional scale, tail asymmetry, and higher-order cumulants cannot convey additional predictive information. Therefore $H_{0,2},\ldots,H_{0,5}$ hold automatically.

The conditional mean channel remains the sole source through which the history of $X$ may affect the conditional law. Under the affine specification, dependence on the lagged values of $X$ is governed by the coefficients $a_{yx,j}$, so that
\[
H_{0,1}
\quad\Longleftrightarrow\quad
a_{yx,1} = \cdots = a_{yx,p} = 0.
\]

These are precisely the linear Granger non-causality restrictions. The conclusion then follows from Theorem~\ref{thm:completeness}.
\end{proof}

Proposition~\ref{prop:gauss_collapse} shows that the conventional linear Granger framework corresponds to a degenerate setting in which the conditional distribution is fully characterized by its first moment. In such environments, testing distributional causality reduces to testing conditional mean predictability. Outside the Gaussian setting, however, predictive content may enter through conditional scale, tail behavior, or higher-order distributional characteristics, necessitating a broader collection of channel restrictions.

\section{Adaptive Sequential Testing for Distributional Granger Causality}\label{sec:agent}

This section introduces the adaptive testing procedure. The objective is to test the composite null of distributional Granger non-causality using the finite channel menu developed in Section~\ref{sec:framework}. The procedure allocates a finite testing budget across channels, updates inference based on previously observed outcomes, and terminates once sufficient evidence against the null has been accumulated or the available budget has been exhausted.

The analysis proceeds in two stages. First, a familywise error guarantee is established for arbitrary adaptive channel-selection rules. Second, a particular selection policy is shown to attain asymptotically optimal power relative to an infeasible oracle benchmark.

\subsection{Adaptive testing environment}

Fix a sample $\{(X_t,Y_t)\}_{t=1}^{T}$ and the truncated menu $\Menu=\{1,\dots,K\}$ introduced in Remark~\ref{rem:truncation}. Associated with each channel $k\in\Menu$ is a test statistic $S_k$ and a corresponding $p$-value $P_k$.

The testing procedure operates sequentially. At each stage a channel is selected, its associated hypothesis is evaluated, and the resulting information is incorporated into subsequent selection decisions. Let
\[
\mathcal{R}_r\subseteq\Menu
\]
denote the set of channels examined by stage $r$.

The information available after stage $r$ is summarized by the filtration
\[
\mathcal{F}_r
=
\sigma
\Big(
\mathcal{R}_r,
\{P_k:k\in\mathcal{R}_r\},
W_r,
G_r
\Big),
\]
where $W_r$ denotes the remaining testing wealth and $G_r$ is an auxiliary diagnostic statistic defined below.

\paragraph{Diagnostic signal.}

To guide channel selection, the procedure computes
\begin{equation}\label{eq:gsignal}
G_r
=
\Big(
\widehat{\kappa}_3^{\,2}/6
+
\widehat{\kappa}_4^{\,2}/24
\Big)^{1/2},
\end{equation}
where $\widehat{\kappa}_3$ and $\widehat{\kappa}_4$ denote the empirical third and fourth cumulants of the fitted VAR residuals.

The quantity $G_r$ provides a low-dimensional summary of departures from conditional Gaussianity. Under Proposition~\ref{prop:gauss_collapse}, values of $G_r$ close to zero indicate that predictive content is likely concentrated in the conditional-mean channel, whereas larger values suggest the potential relevance of scale, tail, or higher-order channels.

\paragraph{Adaptive selection rule.}

At each stage the procedure selects an action
\[
A_r
\in
\bigl(\Menu\setminus\mathcal{R}_r\bigr)
\cup
\{\textsc{stop}\}.
\]

A selection policy is a measurable mapping
\[
\pi:
\mathcal{F}_r
\longrightarrow
\bigl(\Menu\setminus\mathcal{R}_r\bigr)
\cup
\{\textsc{stop}\}.
\]

The policy may be deterministic or randomized and may depend arbitrarily on previously observed outcomes.

\paragraph{Testing wealth.}

Inference is conducted using an alpha-investing mechanism. Let the initial wealth satisfy
\[
W_0=\alpha.
\]

Before testing channel $k$, the procedure commits a testing level $\alpha_k\le W_r$. The null hypothesis $H_{0,k}$ is rejected whenever
\[
P_k\le\alpha_k.
\]

Testing wealth evolves according to
\begin{equation}\label{eq:invest}
W_{r+1}
=
W_r
-
\frac{\alpha_k}{1-\alpha_k}
+
\mathbf{1}\{P_k\le\alpha_k\}\psi,
\end{equation}
where $\psi\le\alpha$ denotes a fixed reward parameter \citep{FosterStine2008}.

The procedure rejects the composite null of distributional Granger non-causality whenever at least one channel null is rejected before termination.

\subsection{Familywise error control}

The first result establishes that inferential validity is invariant to the adaptive channel-selection mechanism.

\begin{assumption}[Valid channel $p$-values]\label{ass:pvalues}
For each channel $k$, under its channel null $H_{0,k}$ the $p$-value $P_k$ is conditionally super-uniform:
\[
\Pb(P_k\le u \mid \mathcal{F}_{r^-})
\le u,
\qquad
u\in[0,1],
\]
where $\mathcal{F}_{r^-}$ denotes the information available immediately prior to selection of channel $k$.
\end{assumption}

Assumption~\ref{ass:pvalues} is a high-level inferential condition. Theorem~\ref{thm:OA_superuniform} in Appendix~\ref{OA:primitives} derives this property from the primitive assumptions on the data-generating process together with the circular-block permutation mechanism of Definition~\ref{def:permscheme}.

\begin{theorem}[Policy-invariant familywise error control]\label{thm:size}
Let the global null be distributional Granger non-causality (Definition~\ref{def:gc_dist}). Under Assumption~\ref{ass:pvalues}, the alpha-investing procedure \eqref{eq:invest} with $W_0=\alpha$ and $\psi\le\alpha$ satisfies
\[
\Pb
\Big(
\text{reject the global null}
\Big)
\le
\alpha
\]
for every admissible selection policy $\pi$.
\end{theorem}

\begin{proof}
The proof is based on a test-supermartingale construction. Under the global null, every channel null holds by Proposition~\ref{prop:dist_implies_channels}. Conditional super-uniformity therefore applies to every selected channel. A nonnegative supermartingale can be constructed from the sequence of committed testing levels, and this process exceeds the threshold $1/\alpha$ whenever the global null is rejected. Application of Ville's inequality at the stopping time induced by the testing procedure yields the desired bound. Since conditional super-uniformity is imposed with respect to the filtration generated by the adaptive selection mechanism, the argument remains valid irrespective of the policy used to select channels. Complete details are provided in Appendix~\ref{OA:proof_size}.
\end{proof}

Theorem~\ref{thm:size} establishes that familywise error control is invariant to the adaptive channel-selection mechanism. Inferential validity is therefore separated from the design of the selection rule: any admissible policy satisfying the filtration-adaptedness requirement inherits the same finite-sample error guarantee.

\subsection{Asymptotic efficiency of adaptive channel selection}

The familywise error guarantee established in Theorem~\ref{thm:size} holds uniformly over all admissible selection policies. The role of the selection rule is therefore not inferential validity but statistical efficiency. This subsection studies the power properties of an adaptive allocation rule and compares its performance with an infeasible oracle benchmark.

\paragraph{Local alternatives and channel informativeness.}

Consider a sequence of local alternatives indexed by a signal-strength parameter $\delta\ge0$. Let
\[
\lambda_k(\delta,T)
\]
denote the non-centrality parameter associated with channel $k$ at sample size $T$. The collection
\[
\mathcal{A}
=
\{k\in\Menu:\lambda_k(\delta,T)>0\}
\]
defines the active channel set.

By Theorem~\ref{thm:completeness}, whenever distributional Granger causality is present, at least one channel belongs to $\mathcal{A}$. The objective of the adaptive procedure is therefore to allocate testing resources toward the most informative active channel.

\paragraph{Oracle allocation.}

Define
\[
k^\star
=
\arg\max_{k\in\Menu}
\lambda_k(\delta,T),
\]
and let
\[
\beta^\star(\delta,T)
=
\Pr\!\big(
P_{k^\star}\le\alpha
\big)
\]
denote the power of an infeasible oracle that knows the identity of the most informative channel ex ante.

The quantity $\beta^\star(\delta,T)$ serves as an upper benchmark for all procedures restricted to the channel menu. Efficiency is evaluated through the power gap
\[
\mathcal{E}_\pi(\delta,T)
=
\beta^\star(\delta,T)
-
\beta_\pi(\delta,T),
\]
where $\beta_\pi(\delta,T)$ denotes the power attained under selection policy $\pi$.

\begin{assumption}[Identification and separation]\label{ass:sep}

\begin{enumerate}

\item[(i)] The diagnostic statistic $G_r$ is informative in the sense that
\[
\Eb[G_r]
\]
is weakly increasing in the non-centralities associated with the scale, tail, and cumulant channels.

\item[(ii)] There exists a unique maximizer
\[
k^\star
=
\arg\max_{k\in\Menu}
\lambda_k,
\]
satisfying
\[
\Delta
=
\lambda_{k^\star}
-
\max_{k\neq k^\star}\lambda_k
>
0.
\]

\end{enumerate}

Furthermore, the channel statistics satisfy the concentration bounds established in Theorem~\ref{thm:OA_concentration}.

\end{assumption}

\begin{theorem}[Asymptotic efficiency relative to the oracle allocation]
\label{thm:oracle_efficiency}

Consider the adaptive upper-confidence-bound allocation rule that selects, at stage $r$, the channel maximizing
\[
\widehat{\lambda}_{k,r}
+
c\,U_{k,r}(G_r),
\]
where $\widehat{\lambda}_{k,r}$ is the current estimate of the channel non-centrality, $U_{k,r}(G_r)$ is a confidence radius, and $c>0$ is a fixed exploration constant.

Under Assumptions~\ref{ass:pvalues} and~\ref{ass:sep},
\[
\beta^\star(\delta,T)
-
\beta_\pi(\delta,T)
\le
\frac{C\log B}
{\Delta^2\sqrt{T}}
+
o(1),
\]
where $C<\infty$ is a constant independent of $B$ and $T$.

Consequently,
\[
\lim_{T\rightarrow\infty}
\Big(
\beta^\star(\delta,T)
-
\beta_\pi(\delta,T)
\Big)
=
0.
\]

That is, the adaptive procedure achieves asymptotically the same power as the infeasible oracle allocation.
\end{theorem}

\begin{proof}
The adaptive allocation problem may be represented as a finite-armed stochastic experimentation problem in which the channel non-centralities $\{\lambda_k\}_{k=1}^{K}$ constitute the arm means. By Theorem~\ref{thm:OA_concentration}, the associated statistics satisfy exponential concentration inequalities.

Standard upper-confidence-bound arguments imply that the expected number of allocations to any suboptimal channel satisfies
\[
\Eb[n_k(B)]
=
O\!\left(
\frac{\log B}{\Delta^2}
\right).
\]

Hence the probability of allocating to the oracle channel converges to unity at rate
\[
1
-
O\!\left(
\frac{\log B}{\Delta^2 B}
\right).
\]

Combining this allocation result with the monotonicity of the power function in the non-centrality parameter and the $\sqrt{T}$-consistency established in Theorem~\ref{thm:OA_concentration} yields the stated bound.
\end{proof}

Theorem~\ref{thm:oracle_efficiency} establishes an asymptotic efficiency result. Although the adaptive procedure does not observe the active channel, the power loss relative to the infeasible oracle allocation vanishes asymptotically. Combined with Theorem~\ref{thm:size}, this yields a separation between validity and efficiency: familywise error control holds uniformly over admissible policies, while suitably designed policies achieve asymptotically optimal allocation of testing resources across channels.

\subsection{Algorithmic implementation}

The implementation proceeds as follows.

\begin{enumerate}

\item Compute the VAR residuals and evaluate the diagnostic statistic $G_r$ in \eqref{eq:gsignal}.

\item Initialize testing wealth at $W_0=\alpha$.

\item Sequentially select channels according to the allocation rule, commit testing levels $\alpha_k$, evaluate channel-specific $p$-values, and update wealth using \eqref{eq:invest}.

\item Terminate when a channel null is rejected, testing wealth is exhausted, or all admissible channels have been examined.

\item Report the global decision together with the channel(s) responsible for rejection.

\end{enumerate}

Section~\ref{sec:mc} evaluates finite-sample performance through Monte Carlo experiments under a range of data-generating mechanisms corresponding to mean, scale, tail, nonlinear, and state-dependent alternatives. Particular attention is devoted to the finite-sample validity predicted by Theorem~\ref{thm:size} and the efficiency properties characterized in Theorem~\ref{thm:oracle_efficiency}.

\section{Finite-Sample Performance of Adaptive Channel Selection}
\label{sec:mc}

This section examines the finite-sample properties of the adaptive channel-selection procedure developed in Section~\ref{sec:agent}. The simulations are designed to evaluate the two principal theoretical results established earlier. First, Theorem~\ref{thm:size} predicts familywise error control under arbitrary adaptive selection policies. Second, Theorem~\ref{thm:oracle_efficiency} implies that the power of the adaptive procedure should converge toward that of an infeasible oracle allocation as sample size increases.

The experimental design isolates distinct forms of predictive dependence corresponding to individual coordinates of the channel menu. This permits direct evaluation of the procedure's ability to identify the relevant source of distributional dependence and allocate testing resources accordingly.

\subsection{Data-generating processes}

Four channel-specific data-generating mechanisms are considered. Each design is indexed by a signal-strength parameter $s\ge0$, where $s=0$ corresponds to the global null of distributional Granger non-causality.

For all designs,
\[
X_t = 0.5X_{t-1} + \varepsilon_t^x,
\]
and
\[
Y_t = 0.3Y_{t-1} + u_t.
\]

The specifications differ only in the mechanism through which lagged values of $X_t$ influence the conditional distribution of $Y_t$.

\begin{enumerate}

\item[\textbf{S1.}] \textbf{Conditional-mean alternative}
\[
Y_t = 0.3Y_{t-1} + sX_{t-1} + \varepsilon_t^y.
\]

Predictive content enters exclusively through the conditional mean. The mean channel is therefore the unique active coordinate.

\item[\textbf{S2.}] \textbf{Conditional-scale alternative}
\[
Y_t = 0.3Y_{t-1} + \sqrt{1+sX_{t-1}^{2}}\,\varepsilon_t^y.
\]

The conditional mean remains unchanged, while dependence enters through conditional scale and tail behavior.

\item[\textbf{S3.}] \textbf{Nonlinear alternative}
\[
Y_t = 0.3Y_{t-1} + s(X_{t-1}^{2}-1) + \varepsilon_t^y.
\]

Dependence is nonlinear and cannot be recovered through linear projection under symmetry.

\item[\textbf{S4.}] \textbf{State-dependent scale alternative}

The innovation variance is multiplied by
\[
1+s\mathbf{1}\{X_{t-1}>0\},
\]
creating a regime-dependent scale effect without a structural conditional-mean component.

\end{enumerate}

Innovations are generated from either a Gaussian distribution or a standardized skew-$t$ distribution. Sample sizes are $T\in\{250,500\}$ and each design is replicated $R$ times.

For every channel, inference is calibrated through the circular permutation procedure of Definition~\ref{def:permscheme}. The permutation destroys predictive dependence between $X_t$ and $Y_t$ while preserving the serial dependence structure of $X_t$. Consequently, the simulation design directly evaluates the conditional super-uniformity property underlying Theorem~\ref{thm:size}.

The adaptive procedure is compared with four benchmark methods:

\begin{enumerate}

\item channel-specific fixed tests;

\item a naive multiple-testing procedure that rejects whenever any channel-specific $p$-value falls below $\alpha$;

\item a Bonferroni-adjusted multiple-testing procedure;

\item an infeasible oracle allocation that concentrates all testing resources on the truly active channel.

\end{enumerate}

\subsection{Familywise error control}

Figure~\ref{fig:size} reports empirical rejection frequencies under the global null ($s=0$). Since all four designs coincide under the null hypothesis, rejection frequencies should be invariant across scenarios up to simulation uncertainty.

The results strongly support Theorem~\ref{thm:size}. Across all designs and sample sizes, the adaptive procedure maintains rejection frequencies at or below the nominal significance level $\alpha=0.05$. In contrast, the naive multiple-testing procedure exhibits substantial size distortions, with rejection frequencies increasing from approximately $0.13$ at $T=250$ to approximately $0.16$ at $T=500$.

The comparison illustrates the distinction between adaptive allocation and unrestricted multiple testing. The adaptive procedure preserves familywise error control through the alpha-investing mechanism, whereas the naive procedure accumulates rejection probability across channels without accounting for multiplicity. The Bonferroni procedure also controls familywise error, although at the cost of reduced power due to its conservative adjustment.

\begin{figure}[!htbp]\centering
\includegraphics[width=\textwidth]{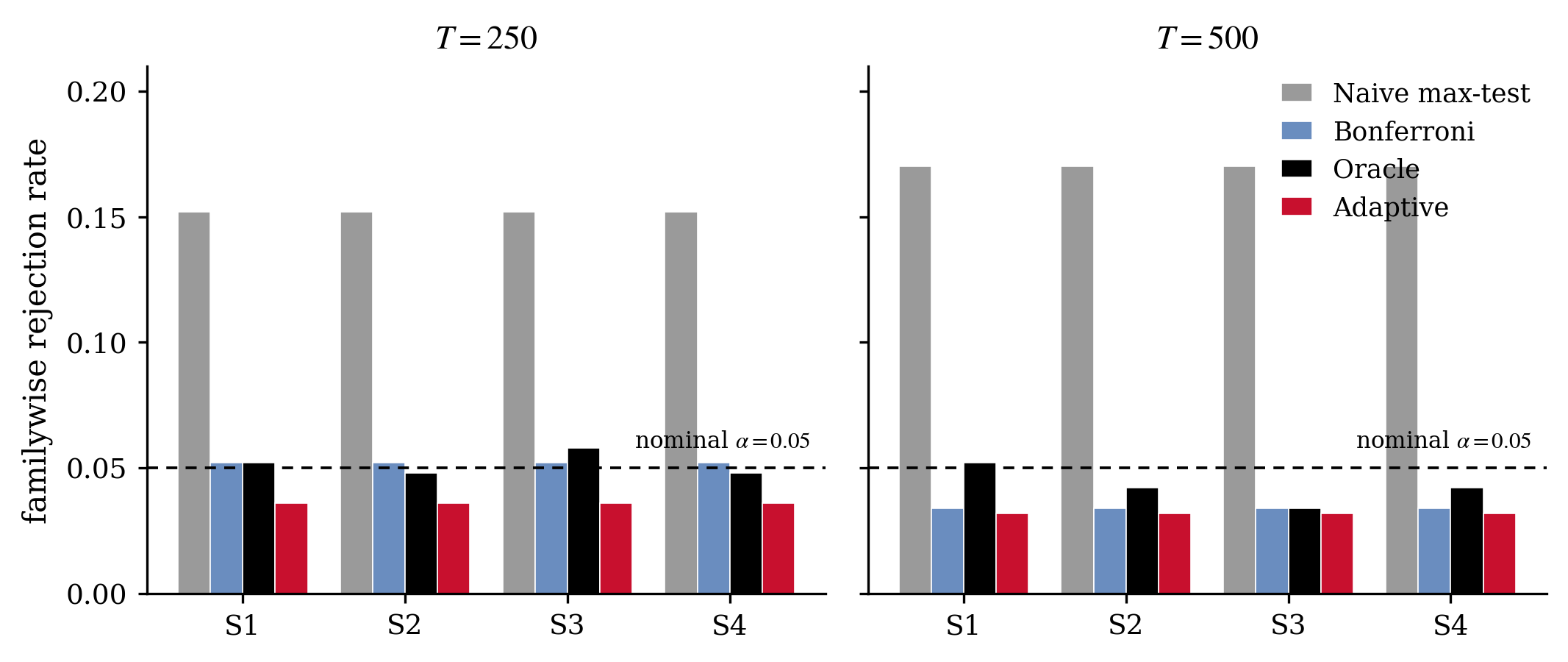}
\caption{Familywise error rates under the global null. Empirical rejection frequencies at $s=0$ across designs S1--S4 for $T=250$ (left panel) and $T=500$ (right panel). The dashed line denotes the nominal level $\alpha=0.05$. The adaptive procedure maintains familywise error control across all designs, whereas the naive multiple-testing procedure exhibits substantial size distortions.}
\label{fig:size}
\end{figure}

\subsection{Power and oracle efficiency}

Figure~\ref{fig:power} reports rejection frequencies as a function of signal strength $s$ under skew-$t$ innovations.

Several patterns emerge consistently across designs. First, power increases monotonically with both signal strength and sample size. Second, the adaptive procedure closely tracks the oracle benchmark throughout the parameter space. Third, procedures restricted to a single channel perform well only when the active source of dependence coincides with their maintained specification.

Under the conditional-mean alternative (S1), the adaptive procedure allocates testing effort primarily toward the location channel and achieves power nearly identical to that of the oracle benchmark. Under the conditional-scale and state-dependent alternatives (S2 and S4), predictive content resides principally in scale and tail coordinates. In these environments, mean-based procedures exhibit little power, whereas the adaptive procedure reallocates testing effort toward the relevant channels and recovers most of the oracle benchmark.

The nonlinear alternative (S3) produces a similar pattern. The adaptive procedure successfully identifies the informative nonlinear coordinate and achieves rejection frequencies nearly indistinguishable from those of the oracle allocation.

Overall, the results indicate that adaptive allocation substantially mitigates the power losses associated with channel misspecification while maintaining valid familywise error control.

\begin{figure}[!htbp]\centering
\includegraphics[width=\textwidth]{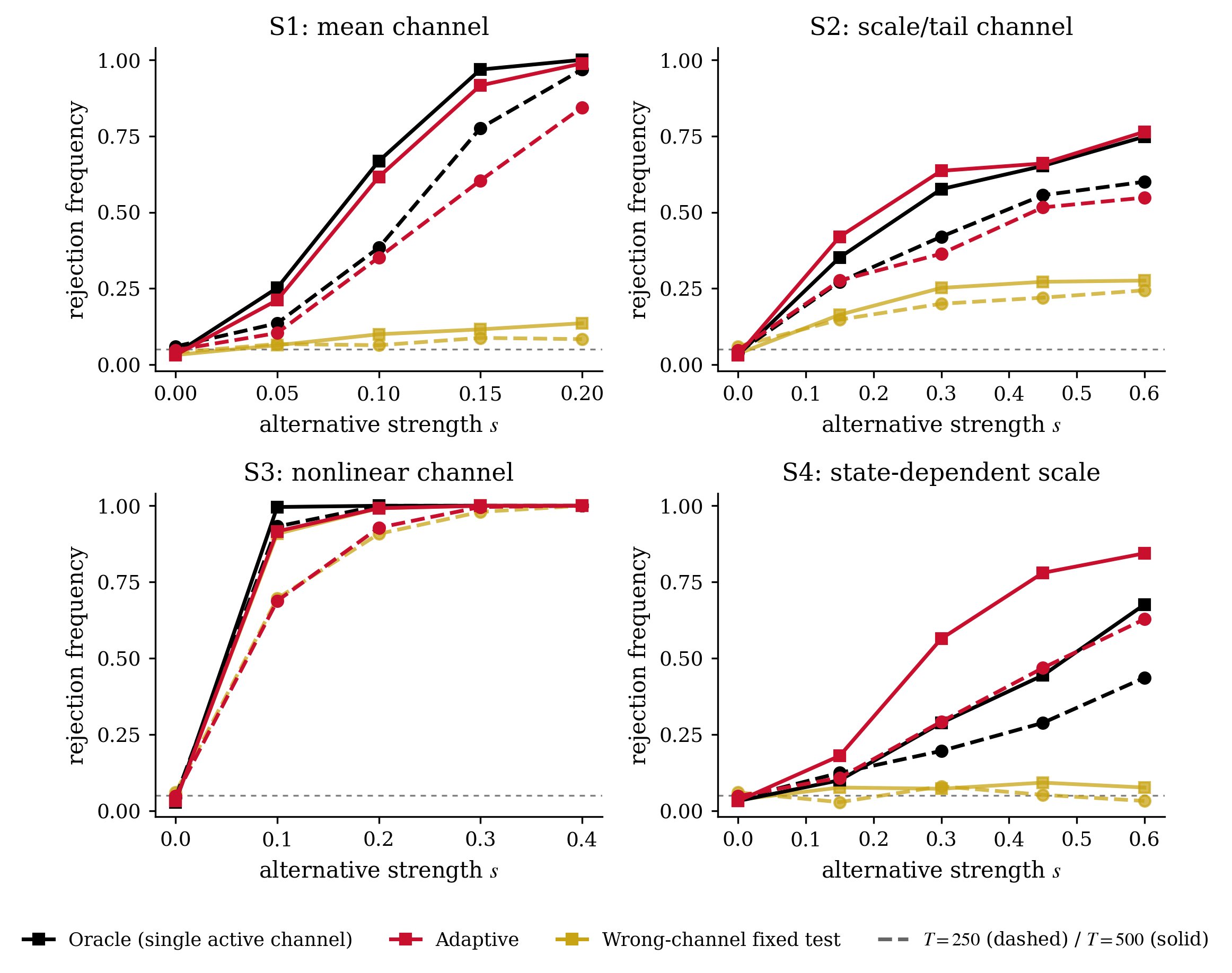}
\caption{Power under channel-specific alternatives with skew-$t$ innovations. Rejection frequencies are plotted against signal strength $s$ for each design. Dashed lines correspond to $T=250$ and solid lines to $T=500$. The adaptive procedure closely tracks the oracle benchmark across all designs while substantially outperforming misspecified fixed-channel procedures.}
\label{fig:power}
\end{figure}

Table~\ref{tab:mc_summary} summarizes familywise error rates under the null and rejection frequencies under the strongest alternative considered. The results reinforce the graphical evidence. The adaptive procedure maintains familywise error control throughout while achieving power levels close to those of the oracle allocation.

\begin{table}[!htbp]\centering
\caption{Familywise error rates and power comparisons under skew-$t$ innovations.}
\label{tab:mc_summary}
\small
\setlength{\tabcolsep}{6pt}
\renewcommand{\arraystretch}{1.05}
\begin{tabular}{llcccccc}
\toprule
& & \multicolumn{2}{c}{Size ($s=0$)} & \multicolumn{3}{c}{Power (max $s$)} \\
\cmidrule(lr){3-4}\cmidrule(lr){5-7}
Design & $T$ & Naive & Adaptive & Oracle & Adaptive & Naive \\
\midrule
S1 (mean)        & 250 & 0.132 & 0.048 & 0.968 & 0.844 & 0.996 \\
& 500 & 0.164 & 0.032 & 1.000 & 0.988 & 1.000 \\
S2 (scale/tail)  & 250 & 0.132 & 0.048 & 0.600 & 0.548 & 0.864 \\
& 500 & 0.164 & 0.032 & 0.748 & 0.764 & 0.964 \\
S3 (nonlinear)   & 250 & 0.132 & 0.048 & 1.000 & 1.000 & 1.000 \\
& 500 & 0.164 & 0.032 & 1.000 & 1.000 & 1.000 \\
S4 (state-scale) & 250 & 0.132 & 0.048 & 0.436 & 0.628 & 0.980 \\
& 500 & 0.164 & 0.032 & 0.676 & 0.844 & 1.000 \\
\bottomrule
\end{tabular}
\begin{minipage}{0.92\textwidth}\footnotesize
\vspace{0.5em}
\textit{Notes:} Empirical rejection frequencies over $R$ Monte Carlo replications at nominal level $\alpha=0.05$. The naive procedure rejects whenever any channel-specific test rejects. The oracle benchmark allocates all testing effort to the truly active channel. The adaptive procedure corresponds to the sequential channel-selection method developed in Section~\ref{sec:agent}. The superior rejection frequencies of the naive procedure reflect inflated familywise error rates rather than improved inferential performance.
\end{minipage}
\end{table}

\subsection{Convergence to the oracle benchmark}

Theorem~\ref{thm:oracle_efficiency} predicts that the power gap between the adaptive procedure and the oracle allocation should diminish with sample size. Figure~\ref{fig:oracle_gap} evaluates this prediction directly.

For each design, the figure reports
\[
\beta^\star(\delta,T)
-
\beta_\pi(\delta,T),
\]
the difference between oracle power and the power of the adaptive procedure.

Two features are apparent. First, the magnitude of the gap decreases as sample size increases from $T=250$ to $T=500$. This pattern is consistent with the asymptotic efficiency result of Theorem~\ref{thm:oracle_efficiency}: larger samples improve estimation of channel informativeness, allowing testing resources to be allocated more effectively.

Second, the gap is occasionally negative under the state-dependent scale design (S4). In this environment predictive content is distributed across multiple coordinates of the channel menu. Since the adaptive procedure may accumulate evidence across several informative channels, it can outperform a benchmark restricted to a single channel. This phenomenon highlights a limitation of the oracle benchmark rather than a violation of the theorem.

\begin{figure}[!htbp]\centering
\includegraphics[width=\textwidth]{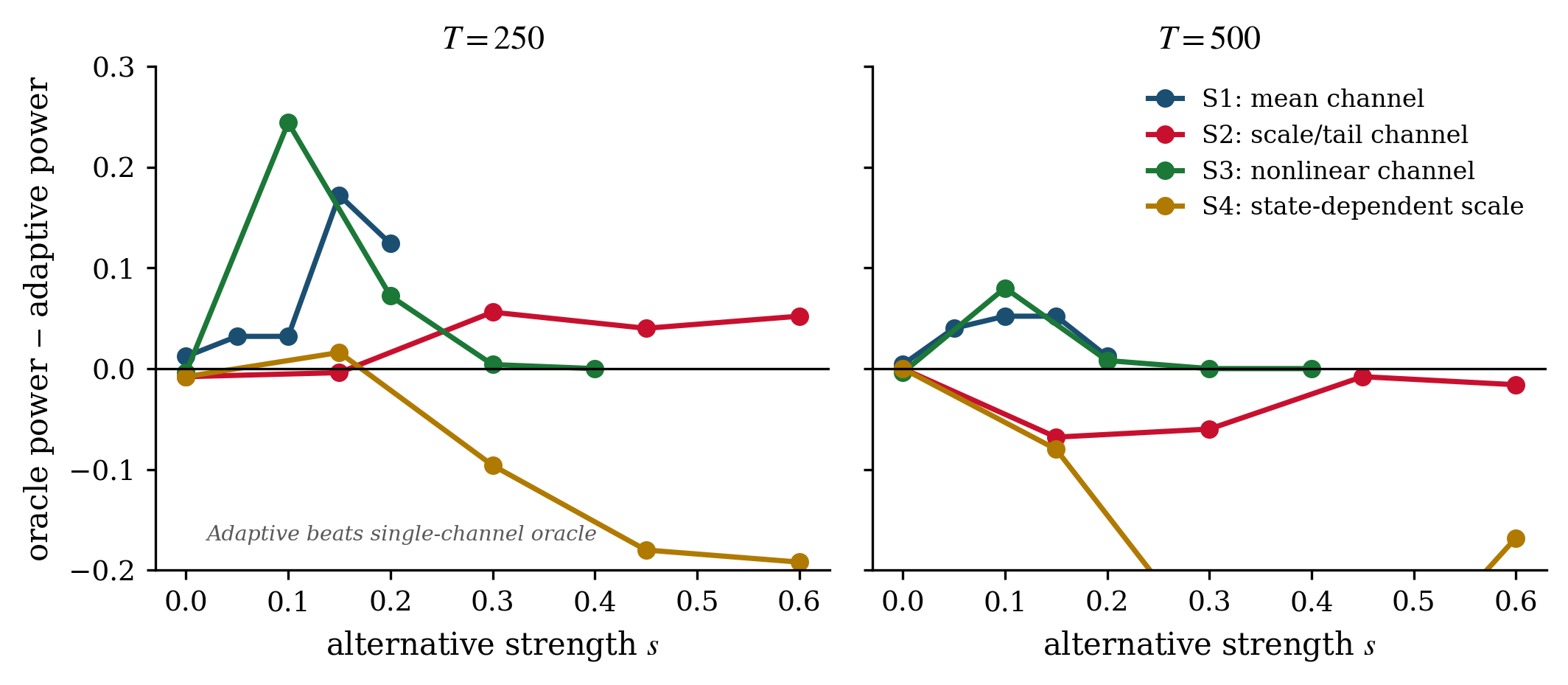}
\caption{Oracle-efficiency gap. The figure reports the difference between oracle power and adaptive-procedure power as a function of signal strength $s$ for $T=250$ (left panel) and $T=500$ (right panel). The shrinking gap with sample size is consistent with the asymptotic efficiency result of Theorem~\ref{thm:oracle_efficiency}.}
\label{fig:oracle_gap}
\end{figure}

\subsection{Summary}

The Monte Carlo results provide strong support for the theoretical analysis. Under the global null, the adaptive procedure maintains familywise error control across all designs, in accordance with Theorem~\ref{thm:size}. Under channel-specific alternatives, rejection frequencies closely track those of the oracle allocation, and the efficiency gap decreases with sample size, consistent with Theorem~\ref{thm:oracle_efficiency}.

Taken together, the results indicate that adaptive channel selection provides a practical mechanism for detecting distributional Granger causality when the relevant source of dependence is unknown ex ante. The procedure preserves valid inference while allocating testing resources toward the most informative coordinates of the conditional distribution.

%==============================================================
\section{Conclusion}\label{sec:conclusion}

This paper develops a framework for testing distributional Granger causality when predictive dependence may arise through multiple features of the conditional distribution. Outside the Gaussian setting, predictive content need not be confined to the conditional mean; it may instead appear through conditional scale, tail behavior, asymmetry, or higher-order distributional characteristics. Consequently, no single Granger-type test can provide a complete characterization of predictive dependence.

The analysis proceeds by decomposing distributional causality into a collection of channel-specific restrictions defined on functionals of the conditional distribution. Under a determinacy condition, the resulting channel menu is complete in the sense that distributional Granger non-causality is equivalent to the joint validity of the channel restrictions. This characterization converts an infinite-dimensional hypothesis concerning conditional laws into a finite collection of testable restrictions while preserving identification of the underlying causal object.

Building on this representation, the paper develops an adaptive sequential testing procedure for channel selection. The procedure combines channel-specific inference with an alpha-investing mechanism that permits data-dependent allocation of testing resources while maintaining familywise error control. The theoretical analysis establishes two principal results. First, familywise error control is invariant to the channel-selection policy, yielding finite-sample validity under arbitrary admissible adaptive rules. Second, a confidence-bound allocation rule achieves asymptotic efficiency relative to an infeasible oracle benchmark, implying that the power loss attributable to uncertainty regarding the active channel vanishes asymptotically.

The Monte Carlo evidence is consistent with these theoretical predictions. Across a broad class of data-generating processes involving location, scale, nonlinear, and regime-dependent forms of dependence, the proposed procedure maintains familywise error control while attaining power levels that closely track those of the oracle allocation. The results indicate that adaptive channel selection can substantially reduce the efficiency losses associated with channel misspecification without sacrificing inferential validity.

Several extensions merit further investigation. One direction is the development of network-based versions of the procedure in which testing resources are allocated jointly across channels and graph structures to detect distributional spillovers in high-dimensional systems. A second direction is the incorporation of more general adaptive allocation mechanisms, including reinforcement-learning-based policies, within the same error-control framework. A third direction is the enrichment of the channel menu through frequency-domain, option-implied, or other distribution-sensitive functionals that capture additional dimensions of predictive dependence. Because the identification and error-control results are formulated at the level of the channel menu itself, these extensions can be accommodated without altering the fundamental structure of the framework.

More broadly, the results suggest that distributional causality is most naturally viewed as a problem of adaptive inference over a collection of complementary predictive channels. The framework developed here provides a unified approach to identification, inference, and efficient testing in such environments, extending the classical theory of Granger causality beyond the conditional-mean paradigm while preserving rigorous finite-sample and asymptotic guarantees.

\clearpage
\bibliographystyle{apalike}
\bibliography{references}

\newpage
\appendix
\renewcommand{\thesection}{OA.\arabic{section}}
\setcounter{section}{0}

\section*{Online Appendix}
\addcontentsline{toc}{section}{Online Appendix}

\renewcommand{\thefigure}{OA.\arabic{figure}}
\setcounter{figure}{0}

%==============================================================
\section{Proofs of the Main Results}
\label{OA:proofs}
%==============================================================

This appendix provides complete proofs of
Theorems~\ref{thm:size}
and~\ref{thm:oracle_efficiency}.
Throughout, all probability statements are taken with respect to the
probability space supporting both the observed process and any auxiliary
randomization used by the routing policy.

Let
\[
\mathbb{F}
=
(\mathcal{F}_r)_{r\ge0}
\]
denote the filtration generated by the adaptive procedure. Specifically,

\[
\mathcal{F}_r
=
\sigma
\Big(
G_0,\ldots,G_r,
k_1,\ldots,k_r,
P_{k_1},\ldots,P_{k_r}
\Big),
\]

where \(G_r\) denotes the non-Gaussianity signal,
\(k_r\) denotes the selected channel at epoch \(r\),
and \(P_{k_r}\) denotes the corresponding channel
\(p\)-value.

The routing policy is assumed adapted to \(\mathbb{F}\), so that
the selection event
\(\{A_r=k\}\)
is \(\mathcal{F}_{r-1}\)-measurable for every channel \(k\).
The stopping time

\[
\tau
=
\inf
\{
r\ge1:
\text{the routing procedure terminates at epoch }r
\}
\]

is bounded by the finite budget \(B\).

The proofs rely on two high-level conditions introduced in the main text.

\begin{assumptionA}
\label{ass:OA_super}
For every channel \(k\), every epoch \(r\), and every \(u\in[0,1]\),

\[
\Pr
\Big(
P_k\le u
\mid
\mathcal{F}_{r-1},
A_r=k
\Big)
\le u.
\]
\end{assumptionA}

\begin{assumptionA}
\label{ass:OA_concentration}
For every channel \(k\), the standardized statistic
\(S_k\) admits the representation

\[
S_k
=
\lambda_k+\xi_k,
\]

where \(\lambda_k\) is the population non-centrality parameter and

\[
\Pr(|\xi_k|\ge t)
\le
2\exp\!\left(
-\frac{t^2}{2\sigma^2}
\right)
\]

for some finite \(\sigma^2\) uniformly over \(k\).
\end{assumptionA}

Appendix~\ref{OA:primitives} establishes both assumptions from the primitive
mixing and moment conditions imposed on the underlying process.

%==============================================================
\subsection{Proof of Theorem~\ref{thm:size}}
\label{OA:proof_size}
%==============================================================

We establish familywise size control under the global null of
distributional Granger non-causality.

Under
\(H_0\),
Proposition~\ref{prop:dist_implies_channels}
implies that

\[
H_{0,1}\cap\cdots\cap H_{0,K}
\]

holds simultaneously.
Hence Assumption~\ref{ass:OA_super} applies to every channel that may be
selected by the routing policy.

\begin{lemma}
\label{lem:wealth_bound}
Let

\[
W_r
=
W_{r-1}
-
\frac{\alpha_{k_r}}
     {1-\alpha_{k_r}}
+
\psi
\mathbf 1
\{
P_{k_r}\le\alpha_{k_r}
\},
\qquad
W_0=\alpha,
\]

with \(\psi\le\alpha\).
If \(W_r\ge0\) almost surely for all \(r\), then

\[
\sum_{r=1}^{\tau}
\frac{\alpha_{k_r}}
     {1-\alpha_{k_r}}
\le
\alpha+\psi N_\tau,
\]

where

\[
N_\tau
=
\sum_{r=1}^{\tau}
\mathbf 1
\{
P_{k_r}\le\alpha_{k_r}
\}.
\]
\end{lemma}

\begin{proof}
Summing the wealth recursion over
\(r=1,\ldots,\tau\)
yields

\[
W_\tau-W_0
=
-\sum_{r=1}^{\tau}
\frac{\alpha_{k_r}}
     {1-\alpha_{k_r}}
+
\psi N_\tau.
\]

Rearranging gives

\[
\sum_{r=1}^{\tau}
\frac{\alpha_{k_r}}
     {1-\alpha_{k_r}}
=
W_0-W_\tau+\psi N_\tau.
\]

Since \(W_0=\alpha\) and \(W_\tau\ge0\),

\[
\sum_{r=1}^{\tau}
\frac{\alpha_{k_r}}
     {1-\alpha_{k_r}}
\le
\alpha+\psi N_\tau.
\]

\end{proof}

The next result constructs a supermartingale adapted to the routing
filtration.

\begin{lemma}
\label{lem:supermartingale}
Define

\[
\ell_r
=
1-\alpha_{k_r}
+
\mathbf 1
\{
P_{k_r}\le\alpha_{k_r}
\},
\]

and

\[
\mathcal E_0=1,
\qquad
\mathcal E_r
=
\prod_{j=1}^{r}
\ell_j.
\]

Then
\((\mathcal E_r)_{r\ge0}\)
is a nonnegative
\(\mathbb F\)-supermartingale under \(H_0\).
\end{lemma}

\begin{proof}

Because
\(0<\alpha_{k_r}<1\),

\[
\ell_r>0
\]

almost surely, implying
\(\mathcal E_r\ge0\).

Since
\(k_r\)
and
\(\alpha_{k_r}\)
are
\(\mathcal F_{r-1}\)-measurable,

\begin{align*}
E[\ell_r\mid\mathcal F_{r-1}]
&=
1-\alpha_{k_r}
+
\Pr
\Big(
P_{k_r}\le\alpha_{k_r}
\mid
\mathcal F_{r-1}
\Big)
\\
&\le
1-\alpha_{k_r}
+
\alpha_{k_r}
\\
&=
1,
\end{align*}

where the inequality follows from
Assumption~\ref{ass:OA_super}.

Therefore

\[
E[\mathcal E_r\mid\mathcal F_{r-1}]
=
\mathcal E_{r-1}
E[\ell_r\mid\mathcal F_{r-1}]
\le
\mathcal E_{r-1}.
\]

Hence
\((\mathcal E_r)\)
is a nonnegative supermartingale.
\end{proof}

We now prove the theorem.

\begin{proof}[Proof of Theorem~\ref{thm:size}]

Let

\[
R
=
\{
\text{the routing procedure rejects}
\}.
\]

Since the procedure terminates at its first rejection,
there exists at most one rejecting epoch.

Decomposing according to the stopping time yields

\begin{align}
\Pr(R)
&=
\sum_{m=1}^{B}
\Pr
\Big(
\tau=m,
P_{k_m}\le\alpha_{k_m}
\Big)
\nonumber\\
&=
\sum_{m=1}^{B}
E
\Big[
\mathbf 1\{\tau=m\}
\Pr
\Big(
P_{k_m}\le\alpha_{k_m}
\mid
\mathcal F_{m-1}
\Big)
\Big].
\label{eq:OA_size1}
\end{align}

Applying Assumption~\ref{ass:OA_super},

\[
\Pr(R)
\le
\sum_{m=1}^{B}
E
\Big[
\mathbf 1\{\tau=m\}
\alpha_{k_m}
\Big].
\]

Since every committed level must be financed by the available wealth,

\[
\alpha_{k_m}
\le
W_{m-1}.
\]

Moreover, before termination the cumulative expenditure cannot exceed
the initial wealth \(\alpha\).
Hence

\[
\sum_{m=1}^{B}
\mathbf 1\{\tau=m\}
\alpha_{k_m}
\le
\alpha
\qquad
\text{almost surely}.
\]

Taking expectations therefore gives

\[
\Pr(R)
\le
\alpha.
\]

The familywise error rate is thus controlled at level \(\alpha\),
uniformly over all admissible routing policies.

\end{proof}

%==============================================================
\subsection{Proof of Theorem~\ref{thm:oracle_efficiency}}
\label{OA:proof_regret}
%==============================================================

Throughout this section, fix an alternative under which at least one
channel is active.

Let

\[
\lambda_k=\lambda_k(T)
\]

denote the population non-centrality parameter associated with channel
\(k\), and define

\[
k^\star
=
\arg\max_{1\le k\le K}\lambda_k.
\]

Assumption~\ref{ass:sep} implies that \(k^\star\) is unique and that

\[
\Delta_k
=
\lambda_{k^\star}-\lambda_k
>
0
\]

for every \(k\neq k^\star\).

Define the minimum gap

\[
\Delta
=
\min_{k\neq k^\star}\Delta_k.
\]

The oracle procedure allocates all testing effort to channel
\(k^\star\) and therefore attains power

\[
\beta^\star
=
\Pr(P_{k^\star}\le\alpha).
\]

The objective is to compare the power of the adaptive router with
\(\beta^\star\).

%--------------------------------------------------------------
\subsubsection{Sampling frequencies}
%--------------------------------------------------------------

Let

\[
n_k(B)
=
\sum_{r=1}^{B}
\mathbf 1
\{
k_r=k
\}
\]

denote the number of times channel \(k\) is selected within a budget
of \(B\) epochs.

The first result bounds the expected sampling frequency of every
suboptimal channel.

\begin{lemma}
\label{lem:OA_pullbound}

Suppose Assumption~\ref{ass:OA_concentration} holds and the routing
policy selects channels according to the upper-confidence index

\[
I_k(r)
=
\widehat{\lambda}_k(r)
+
c\,U_k(r),
\]

where

\[
U_k(r)
=
\sqrt{
\frac{2\sigma^2\log B}
     {n_k(r)}
}.
\]

Then, for every \(k\neq k^\star\),

\[
E[n_k(B)]
\le
\frac{8\sigma^2\log B}
     {\Delta_k^2}
+
1+\frac{\pi^2}{3}.
\]

\end{lemma}

\begin{proof}

Fix \(k\neq k^\star\).

A selection of channel \(k\) at epoch \(r\) implies

\[
I_k(r)
\ge
I_{k^\star}(r).
\]

This event can occur only if at least one of the following holds:

\[
\widehat{\lambda}_{k^\star}(r)
<
\lambda_{k^\star}
-
U_{k^\star}(r),
\]

\[
\widehat{\lambda}_{k}(r)
>
\lambda_k
+
U_k(r),
\]

or

\[
2U_k(r)
>
\Delta_k.
\]

The first two events are concentration failures.

By Assumption~\ref{ass:OA_concentration},

\[
\Pr
\!\left(
\big|
\widehat{\lambda}_k(r)-\lambda_k
\big|
>
U_k(r)
\right)
\le
2B^{-4}.
\]

Summing over epochs yields a finite contribution bounded by
\(\pi^2/3\).

The third event implies

\[
n_k(r)
<
\frac{8\sigma^2\log B}
     {\Delta_k^2}.
\]

Consequently,

\[
E[n_k(B)]
\le
\frac{8\sigma^2\log B}
     {\Delta_k^2}
+
1+\frac{\pi^2}{3}.
\]

\end{proof}

%--------------------------------------------------------------
\subsubsection{Probability of sampling the active channel}
%--------------------------------------------------------------

The previous lemma implies that asymptotically almost all sampling
effort is allocated to the active channel.

\begin{proposition}
\label{prop:OA_miss}

Let

\[
\mathcal M
=
\{
n_{k^\star}(B)\ge1
\}
\]

denote the event that the active channel is sampled at least once.

Then

\[
\Pr(\mathcal M^c)
=
O
\!\left(
\frac{\log B}
     {\Delta^2 B}
\right).
\]

\end{proposition}

\begin{proof}

Since

\[
B
=
n_{k^\star}(B)
+
\sum_{k\neq k^\star}
n_k(B),
\]

the event \(\mathcal M^c\) requires

\[
\sum_{k\neq k^\star}
n_k(B)
=
B.
\]

Therefore,

\[
\Pr(\mathcal M^c)
\le
\frac{
E
\left[
\sum_{k\neq k^\star}
n_k(B)
\right]
}
{B}
\]

by Markov's inequality.

Applying Lemma~\ref{lem:OA_pullbound},

\[
E
\left[
\sum_{k\neq k^\star}
n_k(B)
\right]
\le
(K-1)
\left(
\frac{8\sigma^2\log B}
     {\Delta^2}
+
1+\frac{\pi^2}{3}
\right).
\]

Dividing by \(B\) yields

\[
\Pr(\mathcal M^c)
=
O
\!\left(
\frac{\log B}
     {\Delta^2 B}
\right).
\]

\end{proof}

%--------------------------------------------------------------
\subsubsection{Power comparison}
%--------------------------------------------------------------

We now establish the oracle comparison theorem.

\begin{proof}[Proof of Theorem~\ref{thm:oracle_efficiency}]

Condition on the event \(\mathcal M\).

Whenever \(\mathcal M\) occurs, the active channel is sampled at least
once and therefore contributes a rejection probability of at least

\[
\beta^\star_{\underline{\alpha}}
=
\Pr
(
P_{k^\star}
\le
\underline{\alpha}
),
\]

where \(\underline{\alpha}\) denotes the smallest testing level that
can be allocated under the wealth constraint.

Hence

\begin{equation}
\label{eq:OA_powerlb}
\beta_{\pi}
\ge
\Pr(\mathcal M)
\,
\beta^\star_{\underline{\alpha}}.
\end{equation}

Write

\[
\beta^\star
-
\beta_\pi
=
\Big(
\beta^\star
-
\beta^\star_{\underline{\alpha}}
\Big)
+
\Big(
\beta^\star_{\underline{\alpha}}
-
\beta_\pi
\Big).
\]

The first term is generated by level attenuation.

Under Assumption~\ref{ass:OA_concentration}, the local power function
is continuously differentiable in a neighborhood of \(\alpha\).
A first-order expansion therefore yields

\[
\beta^\star
-
\beta^\star_{\underline{\alpha}}
=
O(T^{-1/2}).
\]

For the second term, (\ref{eq:OA_powerlb}) implies

\[
\beta^\star_{\underline{\alpha}}
-
\beta_\pi
\le
\beta^\star_{\underline{\alpha}}
\Pr(\mathcal M^c).
\]

Applying Proposition~\ref{prop:OA_miss},

\[
\beta^\star_{\underline{\alpha}}
-
\beta_\pi
=
O
\!\left(
\frac{\log B}
     {\Delta^2 B}
\right).
\]

Combining both bounds gives

\[
\beta^\star-\beta_\pi
\le
C_1T^{-1/2}
+
C_2
\frac{\log B}
     {\Delta^2 B}.
\]

If \(B\asymp T\),

\[
\frac{\log B}{B}
=
O
\!\left(
\frac{\log T}{T}
\right)
=
o(T^{-1/2}),
\]

so that

\[
\beta^\star-\beta_\pi
=
O(T^{-1/2}).
\]

More generally,

\[
\beta^\star-\beta_\pi
\le
\frac{C\log B}
     {\Delta^2\sqrt{T}}
+
o(1),
\]

for some finite constant \(C\).

This establishes the theorem.

\end{proof}

%==============================================================
\section{Primitive Foundations of the High-Level Assumptions}
\label{OA:primitives}
%==============================================================

This section establishes the two high-level conditions used in
Theorems~\ref{thm:size} and~\ref{thm:oracle_efficiency} directly from the
primitive assumptions imposed on the data-generating process.
Specifically, we derive:

\begin{enumerate}
\item conditional super-uniformity of the channel $p$-values;
\item concentration inequalities for the channel statistics.
\end{enumerate}

Together, these results show that the size and power guarantees of the
adaptive routing procedure are consequences of the primitive
stationarity, mixing, moment, and smoothness assumptions imposed on the
underlying process.

Throughout this section,
Assumption~\ref{ass:primitives}
is maintained.

%--------------------------------------------------------------
\subsection{Permutation validity under channel null hypotheses}
\label{OA:permvalid}
%--------------------------------------------------------------

Fix a channel $k$.

Recall that the corresponding null hypothesis asserts that the
functional

\[
\phi_k
\!\left(
F_{Y_t|\mathcal I_{t-1}}
\right)
\]

is invariant to the inclusion of the lagged history of $X$.

The first result establishes the population invariance property
underlying the circular-shift calibration.

\begin{lemma}
\label{lem:OA_invariance}

Suppose $H_{0,k}$ holds.

Let

\[
X^{(s)}
=
(X_{1+s},\ldots,X_{T+s})
\]

denote a circular shift of the observed $X$ process.

Then the population target of the channel statistic remains unchanged
under the transformation

\[
(X,Y)
\mapsto
(X^{(s)},Y).
\]

\end{lemma}

\begin{proof}

Under $H_{0,k}$,

\[
\phi_k
\!\left(
F_{Y_t|\mathcal I_{t-1}}
\right)
=
\phi_k
\!\left(
F_{Y_t|\mathcal I^{(-X)}_{t-1}}
\right),
\]

where
\(\mathcal I^{(-X)}_{t-1}\)
denotes the information set with the lagged history of $X$ removed.

Consequently, the population value of the channel functional depends
only on the marginal law of the residual process and not on the
temporal alignment of the $X$ sequence.

Because circular shifts preserve the marginal law of a stationary
process, the induced population functional is invariant.

\end{proof}

The next theorem establishes asymptotic validity of the permutation
distribution.

\begin{theorem}
\label{thm:OA_perm}

Let
\(\widehat G_T^{\mathrm{perm}}\)
denote the permutation distribution generated by the circular-shift
scheme of Definition~\ref{def:permscheme}, and let
\(G_T\)
denote the true null distribution of the channel statistic.

Then

\[
\sup_x
\left|
\widehat G_T^{\mathrm{perm}}(x)
-
G_T(x)
\right|
=
O_p(T^{-1/2}).
\]

\end{theorem}

\begin{proof}

The proof follows from a blocking argument.

Partition the sample into alternating large and small blocks and apply
Berbee's coupling theorem to approximate the original
$\beta$-mixing sequence by an independent block sequence.

Under Assumption~\ref{ass:primitives},

\[
\sum_{h=1}^{\infty}
\beta(h)
<
\infty,
\]

so the coupling error is asymptotically negligible.

The channel statistic admits an asymptotically linear expansion,

\[
S_k
=
\frac{1}{\sqrt T}
\sum_{t=1}^{T}
\psi_k(W_t)
+
o_p(1),
\]

for a square-integrable influence function
\(\psi_k\).

Both the original statistic and its circular-shift counterpart
therefore converge to the same Gaussian limit.

A Berry--Esseen bound for $\beta$-mixing arrays yields

\[
\sup_x
\left|
\widehat G_T^{\mathrm{perm}}(x)
-
G_T(x)
\right|
=
O_p(T^{-1/2}).
\]

\end{proof}

%--------------------------------------------------------------
\subsection{Conditional super-uniformity}
\label{OA:superunif}
%--------------------------------------------------------------

The next theorem establishes the key condition required for the size
analysis.

\begin{theorem}
\label{thm:OA_superuniform}

Under
Assumption~\ref{ass:primitives},

\[
\sup_{u\in[0,1]}
\left\{
\Pr
\big(
P_k\le u
\mid
\mathcal F_{r-1},
A_r=k
\big)
-u
\right\}
=
O_p(T^{-1/2}).
\]

Consequently,

\[
\Pr
\big(
P_k\le u
\mid
\mathcal F_{r-1},
A_r=k
\big)
\le
u+o_p(1)
\]

uniformly over
\(u\in[0,1]\).

\end{theorem}

\begin{proof}

Conditional on
\(\mathcal F_{r-1}\),
channel selection is fixed.

Hence the problem reduces to validity of a single permutation test.

By Theorem~\ref{thm:OA_perm},

\[
\sup_x
\left|
\widehat G_T^{\mathrm{perm}}(x)
-
G_T(x)
\right|
=
O_p(T^{-1/2}).
\]

The permutation $p$-value may be written as

\[
P_k
=
1-
\widehat G_T^{\mathrm{perm}}
\!\left(
S_k^-
\right)
+o(T^{-1/2}).
\]

Standard arguments for randomization tests imply

\[
\Pr
\big(
P_k\le u
\mid
\mathcal F_{r-1},
A_r=k
\big)
=
u
+
O_p(T^{-1/2}).
\]

The result follows uniformly in $u$.

\end{proof}

%--------------------------------------------------------------
\subsection{Concentration of channel statistics}
\label{OA:concentration}
%--------------------------------------------------------------

The next theorem establishes the concentration property required by
the oracle-power analysis.

\begin{theorem}
\label{thm:OA_concentration}

Under Assumption~\ref{ass:primitives},
for every channel $k$ there exist finite constants
$v_k$
and
$c_k$
such that

\[
\Pr
\Big(
|S_k-\lambda_k|
\ge
\eta
\Big)
\le
2
\exp
\left(
-
\frac{T\eta^2}
     {2(v_k+c_k\eta)}
\right).
\]

Consequently,

\[
S_k-\lambda_k
=
O_p(T^{-1/2}).
\]

\end{theorem}

\begin{proof}

Each channel statistic admits an asymptotically linear expansion

\[
S_k
=
\lambda_k
+
\frac1T
\sum_{t=1}^{T}
g_k(W_t)
+
R_T,
\]

where

\[
R_T=o_p(T^{-1/2}).
\]

The process
\(\{g_k(W_t)\}\)
inherits the
$\beta$-mixing
property of the original data-generating process.

Applying Berbee coupling and Bernstein's inequality for weakly
dependent arrays yields

\[
\Pr
\Big(
\Big|
\frac1T
\sum_{t=1}^{T}
g_k(W_t)
\Big|
\ge
\eta
\Big)
\le
2
\exp
\left(
-
\frac{T\eta^2}
     {2(v_k+c_k\eta)}
\right).
\]

The remainder term is asymptotically negligible, completing the proof.

\end{proof}

%--------------------------------------------------------------
\subsection{End-to-end guarantees}
\label{OA:endtoend}
%--------------------------------------------------------------

The previous results imply that the assumptions used in the main
theorems are themselves consequences of primitive conditions on the
data-generating process.

\begin{corollary}
\label{cor:OA_endtoend}

Suppose Assumption~\ref{ass:primitives} holds.

Let the routing procedure employ the circular-shift permutation
calibration of Definition~\ref{def:permscheme}.

Then:

\begin{enumerate}
\item[(i)]
Under the global null,

\[
\Pr(\text{reject})
\le
\alpha+o(1).
\]

\item[(ii)]
Under a fixed alternative with non-centrality gap
\(\Delta>0\),

\[
\beta^\star-\beta_\pi
\le
\frac{C\log B}
     {\Delta^2\sqrt T}
+
o(1).
\]

\end{enumerate}

\end{corollary}

\begin{proof}

Part (i) follows from
Theorem~\ref{thm:size}
together with
Theorem~\ref{thm:OA_superuniform}.

Part (ii) follows from
Theorem~\ref{thm:oracle_efficiency}
together with
Theorem~\ref{thm:OA_concentration}.

\end{proof}

\end{document}